\documentclass[letterpaper, 10 pt, conference]{ieeeconf}  

\IEEEoverridecommandlockouts                              

\usepackage{amsmath} 
\usepackage{amssymb}  

\usepackage{amsthm}
\usepackage{xcolor}
\usepackage{mathtools,stmaryrd}
\usepackage{empheq}
\usepackage[hidelinks]{hyperref}
\usepackage{algorithm}
\usepackage{algpseudocodex}
\let\labelindent\undefined
\usepackage{enumitem}
\usepackage{lipsum}
\usepackage{subcaption}
\usepackage{booktabs}
\usepackage{caption}
\newcommand{\ie}{\emph{i.e.}}
\newcommand{\eg}{\emph{e.g.}}

\newcommand{\ox}{{\mathring{x}}}

\newcommand{\co}{\operatorname{co}}

\newcommand{\immrax}{\texttt{immrax}}

\newcommand{\dbrak}[1]{\llbracket #1 \rrbracket}

\newcommand{\normo}[1]{\dbrak{#1}}

\newcommand{\R}{\mathbb{R}}

\newcommand{\calI}{\mathcal{I}}

\newcommand{\calL}{\mathcal{L}}
\newcommand{\calM}{\mathcal{M}}

\newcommand{\calO}{\mathcal{O}}

\newcommand{\calR}{\mathcal{R}}
\newcommand{\calS}{\mathcal{S}}
\newcommand{\calT}{\mathcal{T}}
\newcommand{\calU}{\mathcal{U}}
\newcommand{\calV}{\mathcal{V}}

\newcommand{\calX}{\mathcal{X}}

\newcommand{\bfI}{\mathbf{I}}

\newcommand{\bfP}{\mathbf{P}}

\newcommand{\ul}[1]{\underline{#1}}

\newcommand{\ulT}{\ul{T}}

\newcommand{\ol}[1]{\overline{#1}}

\newcommand{\olT}{\ol{T}}

\newcommand{\ringz}{\mathring{z}}

\theoremstyle{plain}

\newtheorem{lemma}{Lemma}
\newtheorem{theorem}{Theorem}

\theoremstyle{definition}

\theoremstyle{remark}
\newtheorem{remark}{Remark}

\title{\LARGE \bf
Differentiable Invariant Sets for Hybrid Limit Cycles \\
with Application to Legged Robots
}

\author{
Varun Madabushi, Akash Harapanahalli, Samuel Coogan, Maegan Tucker
\thanks{$^{1}$Varun Madabushi, Akash Harapanahalli, Samuel Coogan, and Maegan Tucker are with the School of Electrical and Computer Engineering, Georgia Institute of Technology, Atlanta, GA, 30332, USA. \{vmadabushi,aharapan,sam.coogan, mtucker\}@gatech.edu}%
}

\newcommand{\fol}{f}
\newcommand{\fcl}{f_{\mathrm{cl}}}
\newcommand{\Deltaol}{\Delta}
\newcommand{\Deltacl}{\Delta_{\mathrm{cl}}}

\begin{document}

\maketitle
\thispagestyle{empty}
\pagestyle{empty}

\begin{abstract}
For hybrid systems exhibiting periodic behavior, analyzing the invariant set containing the limit cycle is a natural way to study the robustness of the closed-loop system.
However, computing these sets can be computationally expensive, especially when applied to contact-rich cyber-physical systems such as legged robots.
In this work, we extend existing methods for overapproximating reachable sets of continuous systems using parametric embeddings to compute a forward-invariant set around the nominal trajectory of a simplified model of a bipedal robot.
Our three-step approach (i) computes an overapproximating reachable set around the nominal continuous flow, (ii) catalogs intersections with the guard surface, and (iii) passes these intersections through the reset map.
If the overapproximated reachable set after one step is a strict subset of the initial set, we formally verify a forward invariant set for this hybrid periodic orbit.
We verify this condition on the bipedal walker model numerically using \immrax, a JAX-based library for parametric reachable set computation, and use it within a bi-level optimization framework to design a tracking controller that maximizes the size of the invariant set.
\end{abstract}

\section{Introduction}

Forward invariance plays a central role in analyzing the robustness of nonlinear dynamical systems \cite{goebel2009hybrid, sastry2013nonlinear}.
However, for hybrid systems (i.e., those that exhibit both continuous dynamics and discrete transitions), accurately characterizing invariant sets is particularly challenging due to the presence of discontinuities and switching behaviors.
These systems arise naturally in applications such as legged locomotion \cite{westervelt2003hybrid}, robotic manipulation with intermittent contact \cite{woodruff2017planning}, and mechanical systems with impacts or frictional constraints \cite{fierro2002hybrid}.

The study of invariant sets and regions of attraction (RoAs) have been approached through a wide range of techniques, including Lyapunov-based methods \cite{khalil2002nonlinear, sastry2013nonlinear}, sum-of-squares (SOS) programming \cite{topcu2008local, manchester2011transverse, chesi2011domain}, and Hamilton–Jacobi (HJ) formulations \cite{mitchell2005time}.
However, the practical applications have remained limited due to the computational complexity of the existing approaches and thus their limitation towards real-world high-dimensional systems.
In 2011, Manchester \cite{manchester2011transverse} presented one of the first algorithmic approaches for computing inner estimates of the regions of attraction of limit cycles of a nonlinear hybrid system using iterative sum-of-squares (SoS) programming.
Their approach formulated Lyapunov inequalities as semidefinite programs and was later demonstrated on the compass-gait walker, a low-dimensional (2-DOF, four states) nonlinear hybrid system \cite{manchester2011regions}.
Despite its theoretical elegance, SoS programming is known to scale poorly with system dimension, as the polynomial degree and number of decision variables grow combinatorially.
Consequently, SoS-based methods are limited to relatively low-dimensional systems, typically below ten state variables, restricting their practical applicability to complex models. 

In 2022, Choi et. al. \cite{choi2022computation} instead cast the RoA computation problem as one of backward reachable set computation within the Hamilton–Jacobi (HJ) reachability framework.
Their work generalized HJ formulations to accommodate hybrid dynamics, successfully recovering larger and more accurate RoA estimates compared to SoS-based approaches.
However, the curse of dimensionality inherent in HJ partial differential equations remains a major bottleneck, as the computational complexity grows exponentially with the dimension of the state space.

To address this challenge of scalability, our work instead quickly computes a parametric overapproximation of the reachable set using a set propagation approach.
Approaches in this category \cite{althoff_set_2021} include generator-based set representations like zonotopes~\cite{althoff_introduction_2015}, or Taylor models bounding the dynamics and/or flow~\cite{bogomolov_juliareach_2019}, and promise scalability to high-dimensional dynamical systems \cite{jafarpour2024Interval}.
In this work, we use linear inclusions to propagate ellipsoidal reachable sets.
This is achieved through \immrax~\cite{harapanahalli2024immrax}, an efficient implementation of interval analysis \cite{jaulin_applied_2001} that exploits the hardware-accelerated and automatic differentiation capabilities of JAX \cite{jax2018github} to enable scalable set propagation.
However, this approach is currently limited to continuous-time systems.
Extending this framework to hybrid systems with discrete impact events addresses new challenges: 1) accurately detecting and enclosing impact times within a continuous set evolution, 2) propagating sets across discrete reset maps that can be highly nonlinear, and 3) maintaining tight enclosures to avoid exponential growth through mode transitions. 

In this work, we build upon the \immrax~ toolbox to address these challenges.
The contributions of this work include the following:
\begin{enumerate}
    \item We extend the theory of parametric reachability to identify conditions under which the parametric reachable set of a hybrid system is forward invariant.
    \item We introduce a procedure using interval analysis to efficiently verify the forward invariance of a given set.
    \item We demonstrate this methodology on a simplified planar walker model and showcase how the differentiability of \immrax~ can be leveraged for control design in a bi-level optimization problem.
\end{enumerate}



\section{Preliminaries}

\subsection{Hybrid System Model}

In this work, we will consider an autonomous hybrid system with one continuous domain and one discrete impact event.
It has been shown in prior work that multi-domain hybrid systems can be treated as an extension of this simpler representation \cite{reher2020algorithmic}. 
Explicitly, the open-loop hybrid control system is represented as:
\begin{subequations} \label{eq:ol_hybridsys}
\begin{align}
    \dot{x} &= \fol(x,u), &\quad  x^-&\in\calX\setminus\calS, \label{eq:ol_hybridsys:cont} \\
    x^+ &= \Deltaol (x^-,v) &\quad x^-&\in\calS, \label{eq:ol_hybridsys:disc}
\end{align}
\end{subequations}
where $x\in\calX\subseteq\R^n$ is the system state, $u\in\calU\subseteq\R^p$ is a control input, $\fol:\calX\times\calU\to\R^n$ defines a parameterized $C^1$ vector field describing the continuous flow until it reaches the guard surface (also known as the switching surface) $\calS \subseteq\calX$, at which point the system undergoes an instantaneous jump from $x^-$ to $x^+$ via the reset map $\Delta:\calX \times \calV \to \calX$ under the discrete input $v\in\calV\subseteq\R^q$. 

Under the control laws $u = \pi^c(x)$ and $v = \pi^d(x^-)$, the closed-loop hybrid control system is represented as:
\begin{subequations} \label{eq:cl_hybridsys}
\begin{align}
  \dot{x}&=\fcl(x) := \fol(x,\pi^c(x)),&\quad x^-&\in \calX \setminus \calS \label{eq:cl_hybridsys:cont} \\
  x^+&=\Deltacl(x^-) := \Deltaol(x^-,\pi^d(x^-))&\quad x^-&\in\calS \label{eq:cl_hybridsys:disc}
\end{align}
\end{subequations}
A solution $x(t)$ of the closed-loop hybrid system is a right-continuous function satisfying~\eqref{eq:cl_hybridsys}, where $x^-(t) = \lim_{\tau\nearrow t} x(\tau)$ and $x^+(t) = \lim_{\tau\searrow t} x(\tau)$.

We assume that the guard surface is given as follows: for some $C^1$ smooth $h:\calX\to\R$,
\begin{align}
    \calS = \{x\in\calX : h(x) = 0, \ \dot{h}(x) < 0\},
\end{align}
where $\dot{h}(x) = \calL_{f_{cl}} h(x)$.
We also notationally set $\{h\sim0\} = \{x\in\calX : h(x) \sim 0\}$ for the relations $\sim\ \in\{=,\leq,\geq,<,>\}$.
Under mild conditions, the closed-loop hybrid system~\eqref{eq:cl_hybridsys} has a unique right-continuous trajectory, which we assume does not have Zeno phenomenon.

A periodic orbit of the closed-loop hybrid system is a solution $x(t)$ of \eqref{eq:cl_hybridsys} that is periodic for some $T$, i.e., $x(t) = x(t + T)$. 
This periodicity condition can also be written using flow notation as:
\begin{align} 
\Delta(\varphi_{T_I(x^*)}(x^*)) = x^*,
\label{eq:periodicity}
\end{align}
where $x^* \in \Delta(\mathcal{S})$ is a fixed point lying on the image of the guard surface, $\varphi_t(x_0)$ denotes in flow notation the solution of our closed-loop continuous dynamics~\eqref{eq:cl_hybridsys:cont} at time $t \in \R_{\geq0}$ given the initial condition $x_0$, and $T := T_I(x^*)$ for the time-to-impact function $T_I: \calX \to \R \cup \{\infty\}$,
\begin{align}
    T_I(x) = \inf \{ t \geq 0 \mid \varphi_t(x) \in \mathcal{S} \},
    \label{eq:time-to-impact}
\end{align}
where we use the convention $\inf\emptyset = \infty$. 
The flow notation can also be used to denote an associated periodic orbit:
\begin{align}
    \mathcal{O} := \{ \varphi_t(x^*) \in \calX \mid 0 \leq t \leq T_I(x^*) = T \}
    \label{eq:orbit}
\end{align}
where \eqref{eq:orbit} is periodic if $x^*$ satisfies \eqref{eq:periodicity}. 
Stability of this periodic orbit is typically analyzed using the Poincar\'e return map \cite{morris2005restricted}. In this work, we will instead use a generalization of this Poincar\'e return map we refer to as the step-to-step map, $F: \calX \to \Delta(\calS)$, defined as:
\begin{align}
    F(x) := \Delta(\varphi_{T_I(x)}(x)).
    \label{eq:step-to-step}
\end{align}
We note that $F$ is only well-defined for points $x\in\calX$ with finite time-to-impact $T_I(x)<\infty$.
Restricting $F$ to the domain $\Delta(\calS)$ would recover the traditional definition of a Poincar\'e map with the Poincar\'e surface $\Delta(\calS)$. This step-to-step map transforms the hybrid system into a discrete-time system:
\begin{align}
    x^+_{k+1} = F(x_k^+), \quad k = 0, 1, \dots
\end{align}
with $x^+ \in \Delta(\mathcal{S})$ being states that lie on the image of the guard (typically the $+$ is used to denote that the states are \textit{post}-impact). Finally, the stability of the fixed point $x^*$ is analyzed by linearizing about $x^*$:
\begin{align}
(x_{k+1}-x^*) \approx A (x_k-x^*):= DF(x^*) (x_k-x^*) \label{eq:linearization}
\end{align}
and computing the eigenvalues of $A$. If the eigenvalues lie strictly within the unit circle ($|\lambda(A)| < 1$), then the hybrid limit cycle is locally asymptotically stable.



\subsection{Function and Dynamics Abstraction via Linear Inclusions and Normotopes}

In this subsection, we review the approach developed in \cite{harapanahalli_LDI_contraction,harapanahalli2025normotope} which use linear inclusions to build a parametric embedding system bounding the reachable set from a norm ball of possibly varying shaping matrix.

\paragraph*{Linear Inclusions}
Suppose we are analyzing a $C^1$ function $g : \R^n\to\R^m$.
A \emph{linear inclusion} encompassing the error behavior of $g$ is an inclusion where for every $x\in\calX$, for prespecified $\ox\in\calX\subseteq\R^n$,
\begin{align} \label{eq:linear_inclusion}
    g(x) - g(\ox) \in \calM (x - \ox),
\end{align}
where $\calM\subseteq\R^{m\times n}$ is a set of matrices.
There are many ways to obtain matrix sets satisfying~\eqref{eq:linear_inclusion}.
For example, if $\overline{\co}\{\frac{\partial g}{\partial x}(x)\} \subseteq \calM$ for every $x\in\calX$ (where $\ol{\co}$ is the closed convex hull), then~\eqref{eq:linear_inclusion} holds for every $x\in\calX$ \cite[Prop. 1]{harapanahalli_LDI_contraction}.

We briefly discuss an algorithmic method from~\cite[Cor. 1]{harapanahalli_LDI_contraction} to obtain $\calM$ using interval analysis on the first partial derivatives of $g$. 
Suppose we are given an interval set $X_1\times\cdots\times X_n\subseteq\R^n$. 
Then the following implication holds, for fixed $\ox\in X_1\times\cdots\times X_n$,
\begin{subequations} \label{eq:mixed_jacobian}
\begin{align} 
    &\frac{\partial g_i}{\partial x_j} (X_1,\dots,X_j,\mathring{x}_{j+1},\dots,\mathring{x}_n) \subseteq [\calM]_{ij}  \\
    &\implies g(x) - g(\mathring{x}) \in [\calM] (x - \mathring{x}),
\end{align}
\end{subequations}
for every $x\in X_1\times\cdots X_n$, where the RHS is evaluated using interval matrix multiplication~\cite{jaulin_applied_2001}.
The set $[\calM]$ in~\eqref{eq:mixed_jacobian} is called the \emph{mixed Jacobian matrix}~\cite[Sec. 2.2.4]{jaulin_applied_2001} since it mixed inputs to the Jacobian matrix between the interval $X$ and the centering point $\ox$.

The expression is further simplified by replacing the interval matrix $[\calM]$ with the convex hull of a finite set of corners satifying $[\calM] \subseteq \co\{M_i\}_i$. For example, if the interval matrix $[\calM]$ has $4$ non-singleton entries, then there are $2^4 = 16$ corners $M_i^x$ to consider.
In \immrax, this procedure is fully automated using automatic differentiation using the \verb|mjacM| transform.

\paragraph*{Normotope Embeddings}

Between two normed vector spaces $(\R^n,\|\cdot\|_a)$, $(\R^m,\|\cdot\|_b)$, let $\|A\|_{a\to b} = \max_{v : \|v\|_a = 1} \|Av\|_b$ denote the induced matrix norm on $\R^{m\times n}$.
For $A\in\R^{n\times n}$, let $\mu(A) = \lim_{h\searrow 0} \frac{1}{h}\|I + hA\|_{a\to a}$ denote the logarithmic norm (also called matrix measure).

Given a norm $\|\cdot\|$ on $\R^n$, a \emph{normotope} is the set
\begin{align}
    \normo{\ox,\alpha,y} := \{x\in\R^n : \|\alpha(x - \ox)\| \leq y\}
\end{align}
where $\ox\in\R^n$ is the \emph{center}, $\alpha\in GL(n)$ is the square invertible $n\times n$ \emph{shape matrix}, and $y>0$ is the \emph{offset}.
The approach described in the previous subsection algorithmically constructs a \emph{linear differential inclusion} (LDI) encompassing the error dynamics of the closed-loop continuous system~\eqref{eq:cl_hybridsys:cont},
\begin{align} \label{eq:LDI}
    \fcl(x) - \fcl(\ox) \in \co\{M_i\}_i (x - \ox),
\end{align}
valid over a particular input set of consideration. 
Following the treatment from \cite{harapanahalli2025normotope}, we build a new dynamical system, called a \emph{parametric embedding system}, by flowing a center $\ox$ according to the nominal dynamics of the closed-loop system, a shaping matrix $\alpha$ according to the adjoint of the linearization about $\ox$, and an offset according to a logarithmic norm bound over the various corners of the linear differential inclusion,
\begin{subequations}\label{eq:normo_embed}
\begin{align} 
    \dot{\ox} &= \fcl(\ox), \\
    \dot{\alpha} &= -\alpha D\fcl(\ox), \\
    \dot{y} &= \max_{i} \left(\mu(\alpha (M_i(t) - D\fcl(\ox)) \alpha^{-1})\right) y.
\end{align}
\end{subequations}
The key property of this embedding system, as shown in~\cite[Thm. 1]{harapanahalli2025normotope} is the following reachable set guarantee: for any $x_0\in\normo{\ox_0,\alpha_0,y_0}$,
\begin{align}
    \varphi_t(x_0) \in \normo{\ox(t),\alpha(t),y(t)},
\end{align}
where $(\ox(t),\alpha(t),y(t))$ is the trajectory of the parametric embedding system~\eqref{eq:normo_embed},
provided the LDI~\eqref{eq:LDI} holds with the matrices $\{M_i(t)\}_i$ for every $x\in\normo{\ox(t),\alpha(t),y(t)}$ for every $t$.



\section{Reachability for Hybrid Periodic Orbits}
This section will derive a theoretical guarantee that the set identified through parametric reachability is a forward invariant set for the hybrid system. We will later demonstrate the practical implementation of obtaining this set in Sec. \ref{sec:implementation}. 


\subsection{Verifying Forward Invariant Tubes via Parametric Reachability}

In Theorem \ref{thm:forward_invariance} below, we apply the following intuitive procedure to verify forward invariant tubes: 
(i) start with a periodic nominal trajectory given by fixed point $x^* = F(x^*)$, impacting the guard at time $T:=T_I(x^*)$, and a normotope centered around its initial condition;
(ii) compute a reachable tube by flowing the normotope embeddding system \eqref{eq:normo_embed} entirely past the guard to ensure every point in the normotope tube resets;
(iii) catalog all intersections between the computed normotope tube and the guard surface into a set $\calI$, and ensure every point in $\calI$ resets into a subset of the initial normotope.



\begin{figure}
    \centering
    \includegraphics[width=\linewidth]{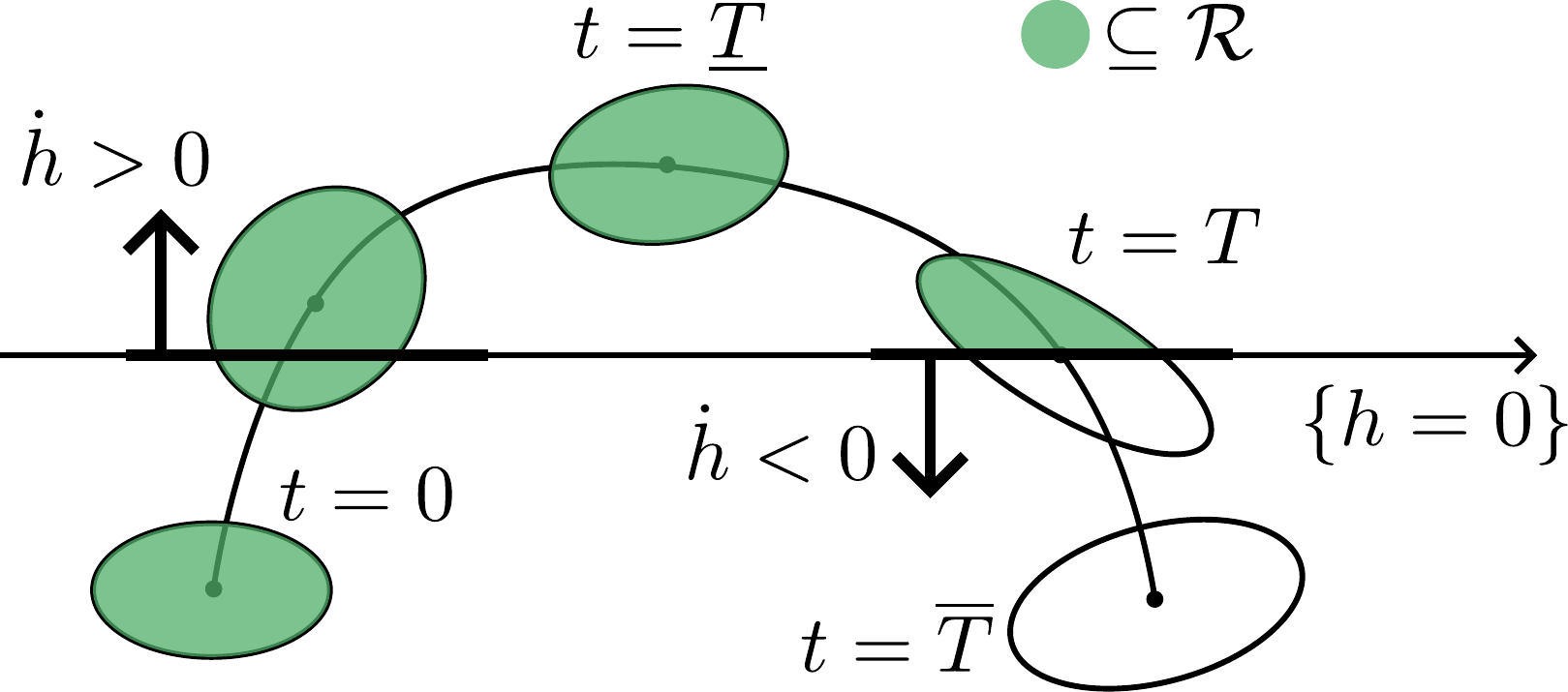}
    \caption{A visualization of the conditions supposed in Theorem \ref{thm:forward_invariance}. While the initial set can start with points in $\{h < 0\}$, we assume there exists a time $\ulT\in[0,\olT)$ where the normotope is contained in $\{h>0\}$.
    Before $\ulT$, trajectories may intersect $\{h=0\}$ as long as $\{\dot{h}>0\}$, avoiding impact with the switching surface $\calS = \{h = 0\}\cap\{\dot{h} < 0\}$.
    By assumption, any intersection with $\{h=0\}$ from $\ulT$ to $\olT$ have $\{\dot{h} < 0\}$, guaranteeing they belong to the guard surface $\calS$.
    Finally, we continue flowing the normotope embedding until the normotope is contained in $\{h < 0\}$ to ensure every point resets.
    }
    \label{fig:theorem1}
\end{figure}

\begin{theorem} \label{thm:forward_invariance}
Let $x^*$ be a fixed point of the step-to-step map $F$ from~\eqref{eq:step-to-step}.
Let $[0,\olT]\ni t\mapsto(\ox(t),\alpha(t),y(t))$ be the embedding system trajectory \eqref{eq:normo_embed} associated to the continuous flow~\eqref{eq:cl_hybridsys:cont} from initial condition $(x^*,\alpha_0,1)$ for some $\olT > T_I(x^*)$,
and for each $t\in[0,\olT]$, let 
\begin{align*}
    \calI_t := \{h = 0\} \cap \normo{\ox(t),\alpha(t),y(t)}.
\end{align*}
Suppose for some $\ulT \in [0,\olT)$ the following conditions hold:
\begin{enumerate}[label=(\alph*)]
    \item \label{hyp:above_guard}$h(\normo{\ox(\ulT),\alpha(\ulT),y(\ulT)}) \subseteq (0,\infty)$ (at time $\ulT$, the set is above $\{h = 0\}$); 
    \item \label{hyp:below_guard}$h(\normo{\ox(\olT),\alpha(\olT), y(\olT)}) \subseteq (-\infty,0)$ (at time $\olT$, the set is below $\{h = 0\}$); 
    \item \label{hyp:transversality_below} $\dot{h}(\calI_t) \subseteq (0,\infty)$ for every $t\in[0,\ulT]$ (transversality from below until time $\ulT$);
    \item \label{hyp:transversality_above} $\dot{h}(\calI_t) \subseteq (-\infty,0)$ for every $t\in[\ulT,\olT]$ (transversality from above after time $\ulT$).
\end{enumerate}
If
\begin{align*}
    \gamma := \sup_{x\in\calI_t,t\in[\ulT,\olT]} \|\alpha_0 (\Delta (x) - x^*)\| \leq 1,
\end{align*}
then \begin{align*}
    \calR = &\bigcup_{t\in[0,\ulT]} \normo{\ox(t),\alpha(t),y(t)} \\
    &\quad \cup \bigcup_{t\in[\ulT,\olT]} (\normo{\ox(t),\alpha(t),y(t)} \setminus \{h < 0\})
\end{align*}
is a forward invariant set for the closed-loop hybrid system~\eqref{eq:cl_hybridsys} containing $\calO$, and the step-to-step map is well-defined on $\calR$, satisfying $F(\calR) \subseteq \normo{x^*,\alpha_0,\gamma}$.
\end{theorem}
\begin{proof}
Let $\calR_1 := \bigcup_{t\in[0,\ulT]} \normo{\ox(t),\alpha(t),y(t)}$, $\calR_2 := \bigcup_{t\in[\ulT,\olT]} \normo{\ox(t),\alpha(t),y(t)} \setminus \{h < 0\}$, noting $\calR = \calR_1\cup\calR_2$.
Then for $x_0\in\calR$, one of the following must be true: (i) there exists a $\tau\in[0,\ulT]$ such that $x_0\in\normo{\ox(\tau),\alpha(\tau),y(\tau)}$, (ii) there exists a $\tau\in[\ulT,\olT]$ such that $x_0\in\normo{\ox(\tau),\alpha(\tau),y(\tau)} \setminus \{h < 0\}$.
In either case, the inclusion $\varphi_t(x_0) \in \normo{\ox(\tau+t),\alpha(\tau+t),y(\tau+t)}$ holds by \cite[Thm. 1]{harapanahalli2025normotope}.

Fix $x_0\in\calR$.
If $x_0\in\calR_1$, Condition \ref{hyp:above_guard} implies at $T_0 = \ulT - \tau$, $h(\varphi_{T_0}(x_0)) > 0$.
If $x_0\in\calR_2$, we have $h(\varphi_{T_0}(x_0)) \geq 0$ for $T_0 = 0$ by assumption.
Condition \ref{hyp:below_guard} implies $h(\varphi_{\olT - \tau}(x_0)) < 0$.
Therefore, by continuity there exists $T_1\in[T_0,\olT - \tau)$ such that $h(\varphi_{T_1}(x_0)) = 0$. 
Since $T_1 + \tau\in[\ulT,\olT]$, Condition \ref{hyp:transversality_above} implies $\dot{h}(\varphi_{T_1}(x_0)) < 0$, and therefore $\varphi_{T_1}(x_0)\in\calS$.

Next, the time-to-impact \eqref{eq:time-to-impact} satisfies $T_I(x_0) \geq \ulT - \tau$, since otherwise, $\dot{h}(\varphi_{T_I(x_0)}(x_0)) < 0$ contradicts Condition \ref{hyp:transversality_below} since $\varphi_{T_I(x_0)}(x_0) \in \calI_{T_I(x_0) + \tau}$.
Thus, $T_I(x_0) + \tau \in [\ulT, T_1] \subseteq[\ulT,\olT]$ by definition of the $\inf$.
Condition \ref{hyp:transversality_above} implies that $\dot{h}(\varphi_{T_I(x_0)}(x_0)) < 0$; therefore, $\varphi_{T_I(x_0)}(x_0) \in \calS$. 

Finally, by definition of $\gamma$, since $\varphi_{T_I(x_0)}(x_0)\in\calI_{T_I(x_0)  + \tau}$, 
\begin{align*}
    \|\alpha_0 (F(x_0) - x^*)\| = \|\alpha_0(\Delta(\varphi_{T_I(x_0)}(x_0)) - x^*)\| \leq \gamma.
\end{align*}
Thus $F(x_0) \in \normo{x^*,\alpha_0,\gamma} \subseteq \normo{x^*,\alpha_0,1} \subseteq \calR$ since $\gamma \leq 1$.
Proceeding by induction, $\calR$ is forward invariant.
\end{proof}


\begin{remark}[Region of Attraction]
While we expect the set $\calR$ from  Theorem~\ref{thm:forward_invariance} to be a region of attraction of $\calO$---especially for the case of ellipsoids with a linear guard---in its current form we have only verified that the set $\calR$ is a forward invariant set containing the orbit.
We leave this as an open question for future work.
\end{remark}

\subsection{Cataloging Guard Intersections: Affine Subspaces and $\ell_2$-Normotopes}

Suppose we consider $\ell_2$-normotopes (which are ellipsoids centered at $\ox$ with shaping matrix $P = \alpha^T\alpha/y^2$), and a guard surface defined by affine function $h(x) = a^T x + b$.
Letting $B\in\R^{n\times(n-1)}$ be a basis for the orthogonal complement to $a$, \emph{i.e.}, $a^TB = 0$ and $B$ has full rank, we can also write $\{h = 0\} = \{Bz + x' : z\in\R^{n-1}\}$, given a distinguished point $x'\in\{h = 0\}$ (\eg, $x' = -\frac{b}{a^Ta} a$).

The following lemma demonstrates how the intersection of an $n$-dimensional $\ell_2$-normotope with an affine subspace can be represented as a $(n-1)$-dimensional $\ell_2$-normotope in the basis $B$, whose $(n-1)\times(n-1)$ shaping matrix is simply the $R$ matrix of the QR decomposition of $\alpha B$.

\begin{lemma}[Slicing an $\ell_2$-normotope] \label{lem:slicing}
Let $\normo{\ox,\alpha,y}_2\subseteq\R^n$ be an $\ell_2$-normotope, $B\in\R^{n\times (n-1)}$ be a basis for a $(n-1)$-dimensional subspace, and $x'\in\R^n$ be an offset vector.
The intersection is characterized as follows,
\begin{align*}
    &\normo{\ox,\alpha,y}_2 \cap \{Bz + x' : z\in\R^{n-1}\} \\
    &= \begin{cases}
        B\normo{\ringz,R,\sqrt{y^2 - r}}_2 + x', & y^2 \geq r \\
        \emptyset & y^2 < r
    \end{cases},
\end{align*}
where $QR = \alpha B$ is the reduced QR-decomposition of $\alpha B$, $\ringz = R^{-1}Q^T\alpha(\ox - x')$, and $r = \|\alpha(\ox - x')\|_2^2 - \|Q^T \alpha(\ox - x')\|_2^2$.
\end{lemma}
\begin{proof}
We first prove the $x'=0$ case. The set corresponds to any $z\in\R^{n-1}$ satisfying 
\begin{align} \label{eq:lem_slicing_cond1}
    z^T(\alpha B)^T (\alpha B)z - 2(\alpha\ox)^T \alpha B z \leq y^2 - \ox^T\alpha^T\alpha\ox.
\end{align}
Letting $\alpha B = QR$ be the reduced QR decomposition ($Q\in\R^{n\times (n-1)}$ with orthonormal columns, and $R\in\R^{(n-1)\times(n-1)}$), we see that $(\alpha B)^T (\alpha B) = R^T Q^T QR = R^TR$. Completing the square, the LHS of \eqref{eq:lem_slicing_cond1} is equal to
\begin{align*}
    z^T R^T R z &- 2 (\alpha\ox)^T QR z = - \ox^T\alpha^TQ Q^T \alpha\ox\\
    &+ (z - R^{-1}Q^T \alpha \ox)^T R^T R (z - R^{-1}Q^T \alpha\ox).
\end{align*}
Adding the offset to the RHS, we obtain the condition
\begin{align*}
    \|R(z - R^{-1}Q^T\alpha\ox)\|_2^2 \leq y^2 - \underbrace{(\|\alpha\ox\|_2^2 - \|Q^T\alpha\ox\|_2^2)}_{=:r}.
\end{align*}
When $y^2 < r$, there are no $z$ satisfying the condition because the LHS is nonnegative (thus, the intersection is empty). When $y^2 \geq r$, a square root of both sides shows $z$ lives in the normotope subset
\begin{align*}
    \normo{R^{-1}Q^T\alpha\ox, R, \sqrt{y^2 - r}}_2 \subseteq\R^{n-1}.
\end{align*}
When $x'\neq 0$, translate the subspace and the normotope so the subspace passes through the origin, apply the previous case to $\ox - x'$, and shift back to complete the proof.
\end{proof}

\begin{theorem} \label{thm:mixed_guard}
Consider the setting of Theorem \ref{thm:forward_invariance} with $\|\cdot\|_2$.
Let $\Delta^B : \R^{n-1}\to\R^n$ denote $z\mapsto\Delta(Bz + \ox(T))$, the reset map as a function of the basis for the affine subspace $\{h = 0\}$, centered at $\ox(T)\in\{h = 0\}$ for $T := T_I(x^*)$.
For every $t\in[\ulT,\olT]$, let $R(t)$, $\ringz(t)$, and $r(t)$ be defined as Lemma \ref{lem:slicing} for the normotope $\normo{\ox(t),\alpha(t),y(t)}_2$, and let $\calT = \{t\in[\ulT,\olT] : y(t)^2 \geq r(t)\}$ denote the times with nonempty intersection.
Let $\{M^\Delta_j(t)\}_{j}\subseteq\R^{n\times (n-1)}$ define a collection of corners satisfying the linear inclusion
\begin{align*}
    \Delta^B(z) - \Delta^B(\ringz(t)) \in \co\{M^\Delta_j(t)\}_j (z - \ringz(t)),
\end{align*}
for every $t\in\calT$ and $z\in\normo{\ringz(t),R(t),\sqrt{y(t)^2 - r(t)}}_2$.
Then 
\begin{align*}
    \calI_t = \begin{cases}
        B\normo{\ringz(t),R(t),\sqrt{y(t)^2 - r(t)}}_2 + \ox(T) & t\in\calT \\
        \emptyset & t\notin\calT
    \end{cases},
\end{align*}
and the gain $\gamma$ from Theorem \ref{thm:forward_invariance} can be bounded as follows,
\begin{align*}
    \gamma  
    &\leq \sup_{t\in\calT}  \max_{j} \Big(\|\alpha_0(\Delta^B(\ringz(t)) - x^*)\|_2 \\
    & \quad \quad + \|\alpha_0 M^\Delta_j(t) R(t)^{-1}\|_{2\to 2} \sqrt{y(t)^2 - r(t)} \Big).
\end{align*}
\end{theorem}
\begin{proof}
By Lemma \ref{lem:slicing},
$\calI_t = \emptyset$ for $t\notin\calT$ and $\calI_t = B\normo{\ringz(t),R(t),\sqrt{y(t)^2 - r(t)}}_2 + \ox(T)$ for $t\in\calT$, thus,
\begin{align*}
    \gamma &\leq \sup_{z\in\normo{\ringz(t),R(t),\sqrt{y(t)^2 - r(t)}}_2,t\in\calT} \|\alpha_0(\Delta(Bz + \ox(T)) - x^*)\|.
\end{align*}
Fix $t\in\calT$ and $z\in\normo{\ringz(t),R(t),\sqrt{y(t)^2 - r(t)}}_2$; then there exists $M\in\co\{M^\Delta_j(t)\}_j$ such that $\Delta^B(z) - \Delta^B(\ringz(t)) = M (z - \ringz(t))$.
Applying the triangle inequality,
\begin{align*}
    &\|\alpha_0(\Delta(Bz + \ox(T)) - x^*\|_2 = \|\alpha_0(\Delta^B(z) - x^*)\|_2 \\
    &\leq \|\alpha_0(\Delta^B(z) - \Delta^B(\ringz(t)))\|_2 + \|\alpha_0(\Delta^B(\ringz(t)) - x^*)\|_2 \\
    &= \|\alpha_0 M (z - \ringz(t))\|_2 + \|\alpha_0(\Delta^B(\ringz(t)) - x^*)\|_2.
\end{align*}
Furthermore, since $\alpha_0 M (z - \ringz(t)) = \alpha_0 M R(t)^{-1}R(t)(z - \ringz(t))$, using the induced matrix norm between $(\R^{n-1},\|\cdot\|_2)$ and $(\R^{n},\|\cdot\|_2)$,
\begin{align*}
    &\|\alpha_0 M (z - \ringz(t))\|_2 \leq \|\alpha_0 M R(t)^{-1}\|_{2\to 2} \|R(t)(z - \ringz(t))\|_2 \\
    &\leq \|\alpha_0 M R(t)^{-1}\|_{2\to 2} \sqrt{y(t)^2 - r(t)} \\
    &\leq \max_{j} \|\alpha_0 M^\Delta_j(t) R(t)^{-1}\|_{2\to 2} \sqrt{y(t)^2 - r(t)}
\end{align*}
using convexity of the induced matrix norm to bound $M\leq\sup_{M^\Delta\in\co\{M_j^\Delta\}_j} \|\alpha_0M^\Delta R(t)^{-1}\|_{2\to 2} = \max_{j} \|\alpha_0M^\Delta_j R(t)^{-1}\|_{2\to 2}$.
\end{proof}

\begin{figure}
    \centering
    \includegraphics[width=\linewidth]{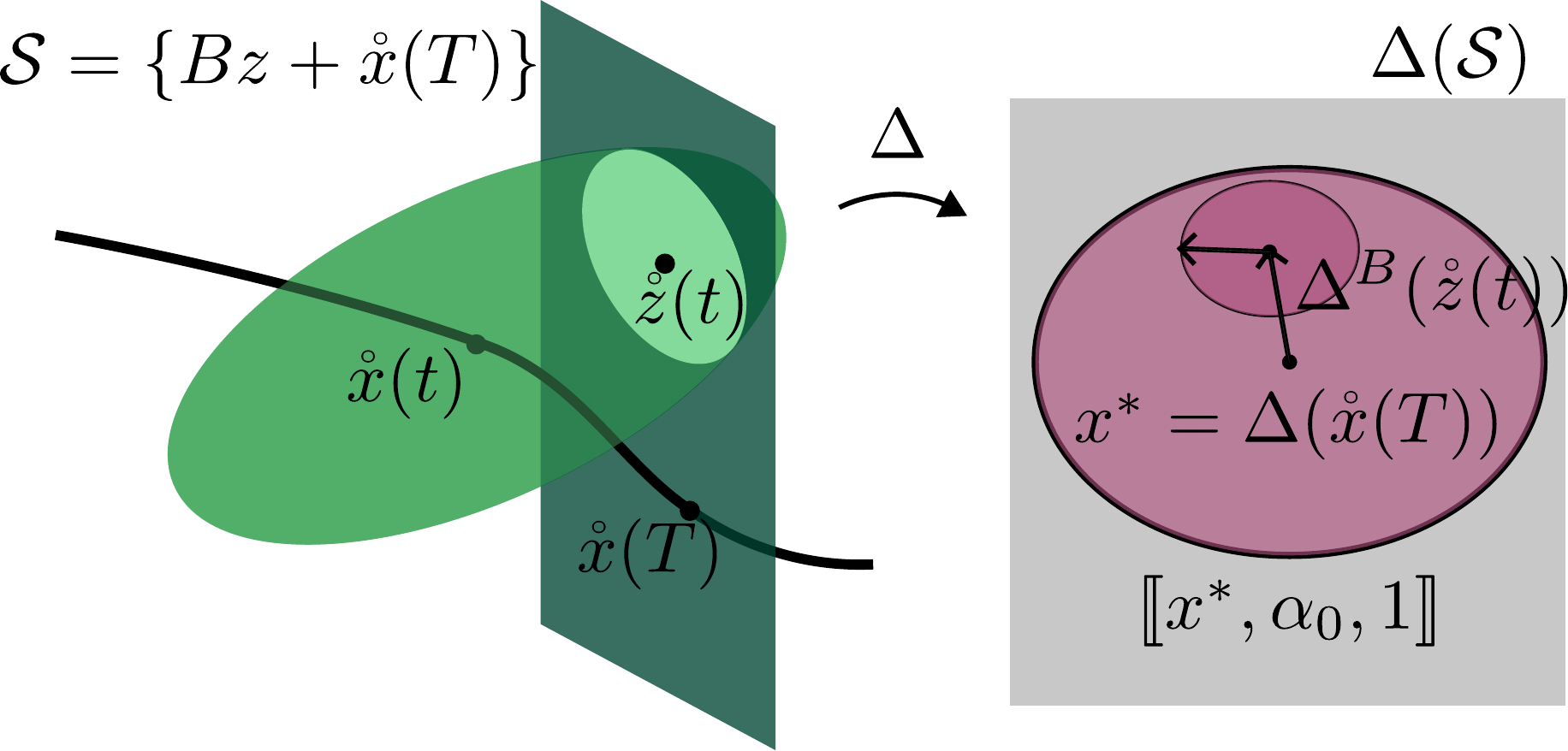}
    \caption{The slicing and bounding procedure used in Lemma \ref{lem:slicing} and Theorem \ref{thm:mixed_guard} is pictured. 
    \textbf{Left}: the $\ell_2$-normotope at time $t$ is sliced with the subspace $\{h=0\}= \{Bz + \ox(T)\}$ to obtain a $(n-1)$-dimensional $\ell_2$-normotope expressed in the basis $B$ (Lemma \ref{lem:slicing}).
    \textbf{Right}: The gain $\gamma$ is bounded using (i) the distance from $x^*$ to the reset center of the sliced normotope $\Delta^B(\ringz(t))$ measured in the initial $\alpha_0$ norm, and (ii) the maximum induced matrix norm of a linear inclusion bounding $\Delta^B(z) - \Delta^B(\ringz(t))$ to obtain a radius in the $\alpha_0$ norm.
    }
    \label{fig:theorem1}
\end{figure}

\section{Algorithmic Computation of Forward Invariant Sets}
\label{sec:implementation}
In this section, we describe our implementation of set-based reachability, which builds upon \immrax{} \cite{harapanahalli2024immrax}.
Our approach promises computational efficiency, scalability, and differentiability, unlocking future applications to higher-dimensional systems.
We model reachable sets as $\ell_2$-normotopes of the form $\normo{\ox,\alpha,y}_2 = \{x : \|\alpha(x - \ox)\| \leq y\}$, which correspond to ellipsoids.



\subsection{Obtaining an Initial Set}
\label{sec:initial_set}
Suppose we are given a closed-loop hybrid system~\eqref{eq:cl_hybridsys} with a fixed point $x^*$ of the step-to-step map~\eqref{eq:step-to-step}, corresponding to a periodic orbit $\mathcal{O}$ with period $T$. 
Before applying Theorem~\ref{thm:forward_invariance}, a natural question is how to obtain a suitable shape matrix $\alpha_0$ to initialize the parametric embedding system, \ie, a suitable initial ellipsoid in the $\ell_2$ case.

To find a suitable shaping matrix $\alpha_0$ defining the initial normotope set $\normo{x^*,\alpha_0,1}$ for Theorem~\ref{thm:forward_invariance}, we apply ideas from contraction theory.
Let $A$ be the linearization of the closed-loop step-to-step system about the fixed point, that is, $A=DF(x^*)$.
Setting $b = \rho(A) + \varepsilon$ for a small constant $\varepsilon>0$, we solve the following semi-definite program to find a norm which decays at rate $b$, which is the minimum contraction rate \cite[Eq. (2.36b)]{bullo2022contraction}
\begin{align}\label{eq:initial_set_SDP}
    \min_{\bfP\succ \bfI} & \quad \operatorname{trace}(\bfP) \quad  \text{s.t. } \quad A^T \bfP A \preceq b^2 \bfP.
\end{align}
Setting $\alpha_0 = \bfP^{1/2}$ using the positive semi-definite matrix square root, the resulting norm $\|\alpha_0 x\|_2$ matches the optimal norm $\sqrt{x^T \bfP x}$. 



\subsection{Computing the Continuous Reachable Tube}
\label{sec:reachable_tube}
Once we have an initial set $\normo{x^*,\alpha_0,1}_2$, we obtain a reachable tube by forward-simulating the dynamical embedding system~\eqref{eq:normo_embed}.
We simulate the parametric embedding system using Euler integration, generating a normotope trajectory $\normo{\ox_i, \alpha_i, y_i}$.
We simulate the trajectory until the nominal trajectory intersects the guard surface, and then further until the normotope fully passes through the guard surface, corresponding to the time $\olT$.



\subsection{Accounting for the Guard Surface and Reset Map}
\label{sec:guard_and_reset}
To compute the reachable set through the hybrid dynamics, we check intersection of each ellipsoid in the normotope trajectory with the guard surface, yielding a set of intersecting normotopes.
To each of the intersecting normotopes, we apply the linear inclusion of the  reset map and apply Theorem \ref{thm:mixed_guard} to find the upper bound of the overall gain $\gamma$.

\subsection{Resizing the Invariant Tube}
\label{sec:resizing}
The procedure described in Sec. \ref{sec:guard_and_reset} can be used to verify the invariance of a tube defined by its initial shaping matrix, but does not necessarily imply that such a tube is the largest possible one.
Additionally, the initial choice of shaping matrix $\alpha_0$ derived from contraction theory is only accurate up to a scale.
Therefore, we introduce a bisection-based algorithm to increase the size of the tube while preserving its forward invariance. 

Consider a function $\Phi: GL(n)\rightarrow \mathbb{R}$ which maps an initial shaping matrix to the associated post-impact length, as in Theorem \ref{thm:forward_invariance}.
Verifying that the hybrid reachable tube associated with the initial normotope shaped by $\alpha_0$ is forward invariant is equivalent to ensuring that $\Phi(\alpha_0)<1$.
This normotope can be uniformly scaled by dividing $\alpha_0$ by some factor $s$.
We select the scale factor according to the optimization problem
\begin{align}\label{eq:scaling_bisection}
    \max & \quad s \quad  \text{s.t. } \Phi(\alpha_0/s) < 1 .
\end{align}
This amounts to a single-parameter search and can thus be easily solved using a bisection algorithm.

\subsection{Software Implementation}
The implementation of our set-based reachability approach builds upon \immrax~ and is fully compatible with the JAX software ecosystem, enabling parallelism, traceability, and automatic differentiation.

We numerically implement $F(x)$ by simulating the dynamics of \eqref{eq:cl_hybridsys:cont} starting at $x$ until the trajectory intersects the guard surface, using the event-based integrator in \verb|diffrax|~\cite{kidger2021on} to accurately determine the guard intersection point $x^-$.
The reset function \eqref{eq:cl_hybridsys:disc} is then applied to $x^-$, yielding a new point $\Delta(x^-) = F(x)$.
Using the automatic differentiation capability of JAX, we can additionally compute the Jacobian of $F(x)$, used in Sec. \ref{sec:initial_set}.

We also implement the steps described in Sec. \ref{sec:reachable_tube} and \ref{sec:guard_and_reset} using the normotope toolbox in \immrax.
The benefit of our approach is that every step in the reachable set computation can be automatically differentiated with respect to its inputs, which ultimately allows us to compute the gradient of $\gamma$ with respect to system parameters.
This property can be leveraged in control design, as demonstrated in Sec. \ref{sec:tracking_control}.
An implemention of the entire framework is provided as a Python notebook\footnote{\href{https://github.com/dynamicmobility/Hybrid-Invariant-Sets.git}{https://github.com/dynamicmobility/Hybrid-Invariant-Sets.git}}.


\section{Application to Planar Bipedal Walker}
\label{sec:walker}
We apply our methodology on a simple planar bipedal walker.
This model takes inspiration from the rimless-wheel and the compass-gait walker, but instead features a telescoping leg to better model the effect of a knee joint.
Since we are proposing a novel model, we will first introduce our system and propose a closed-loop controller under which the system has a stable periodic limit cycle.


\begin{figure}[t]
    \centering
    \begin{subfigure}{0.32\linewidth}
        \centering
        \includegraphics[width=\linewidth]{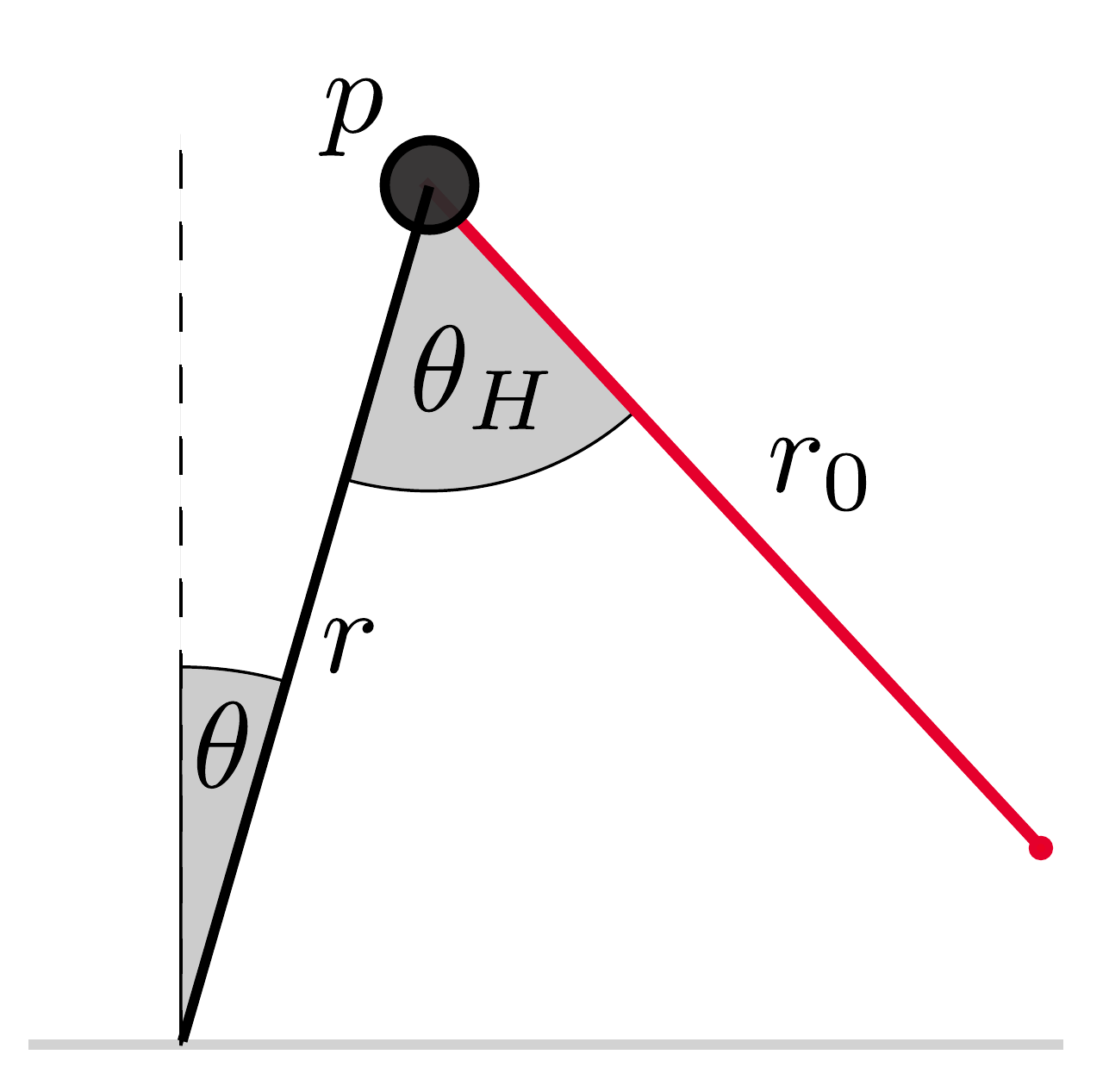}
        \caption{}
        \label{fig:continuous}
    \end{subfigure}\hfill
    \begin{subfigure}{0.32\linewidth}
        \centering
        \includegraphics[width=\linewidth]{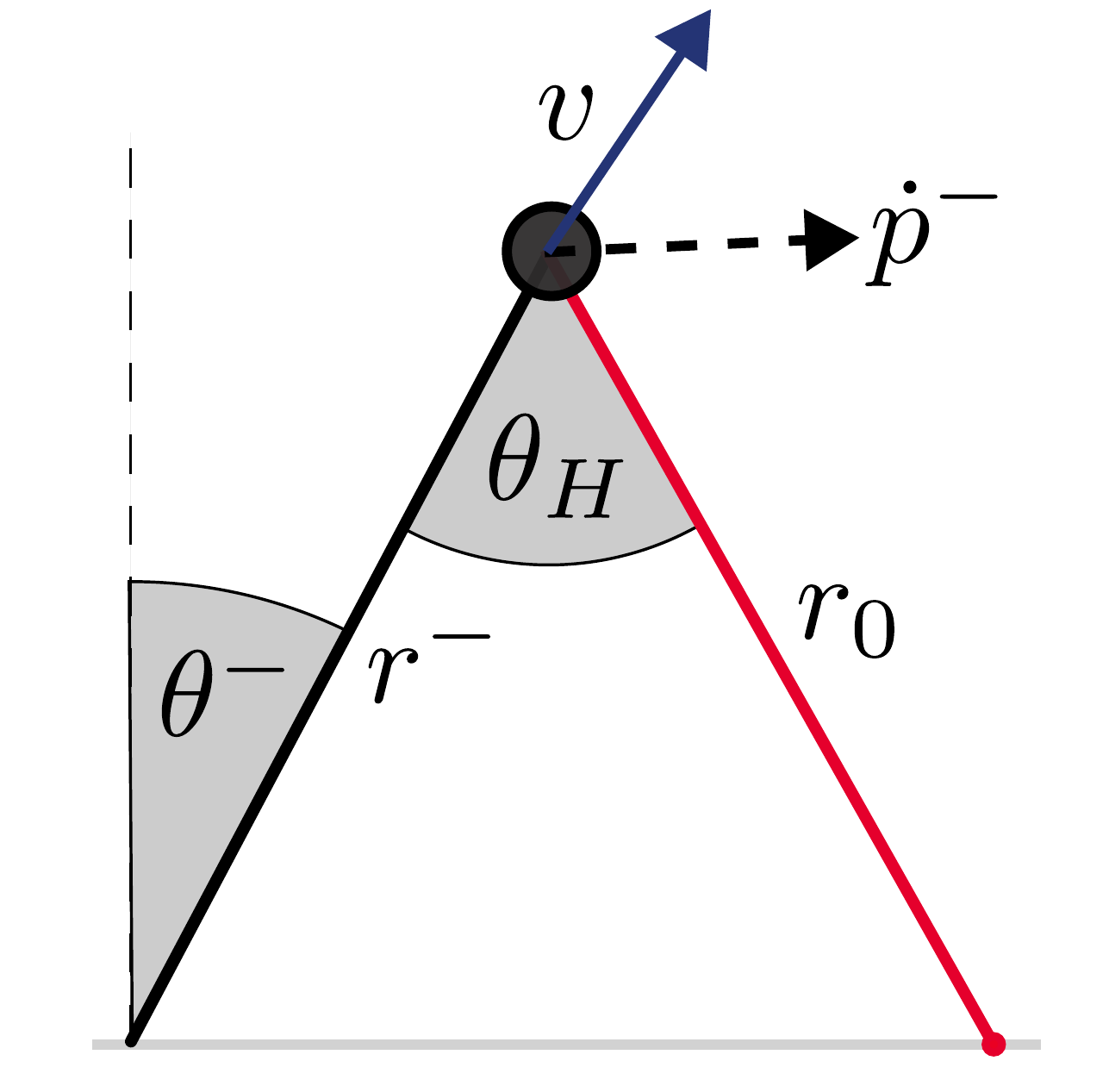}
        \caption{}
        \label{fig:preimpact}
    \end{subfigure}\hfill
    \begin{subfigure}{0.32\linewidth}
        \centering
        \includegraphics[width=\linewidth]{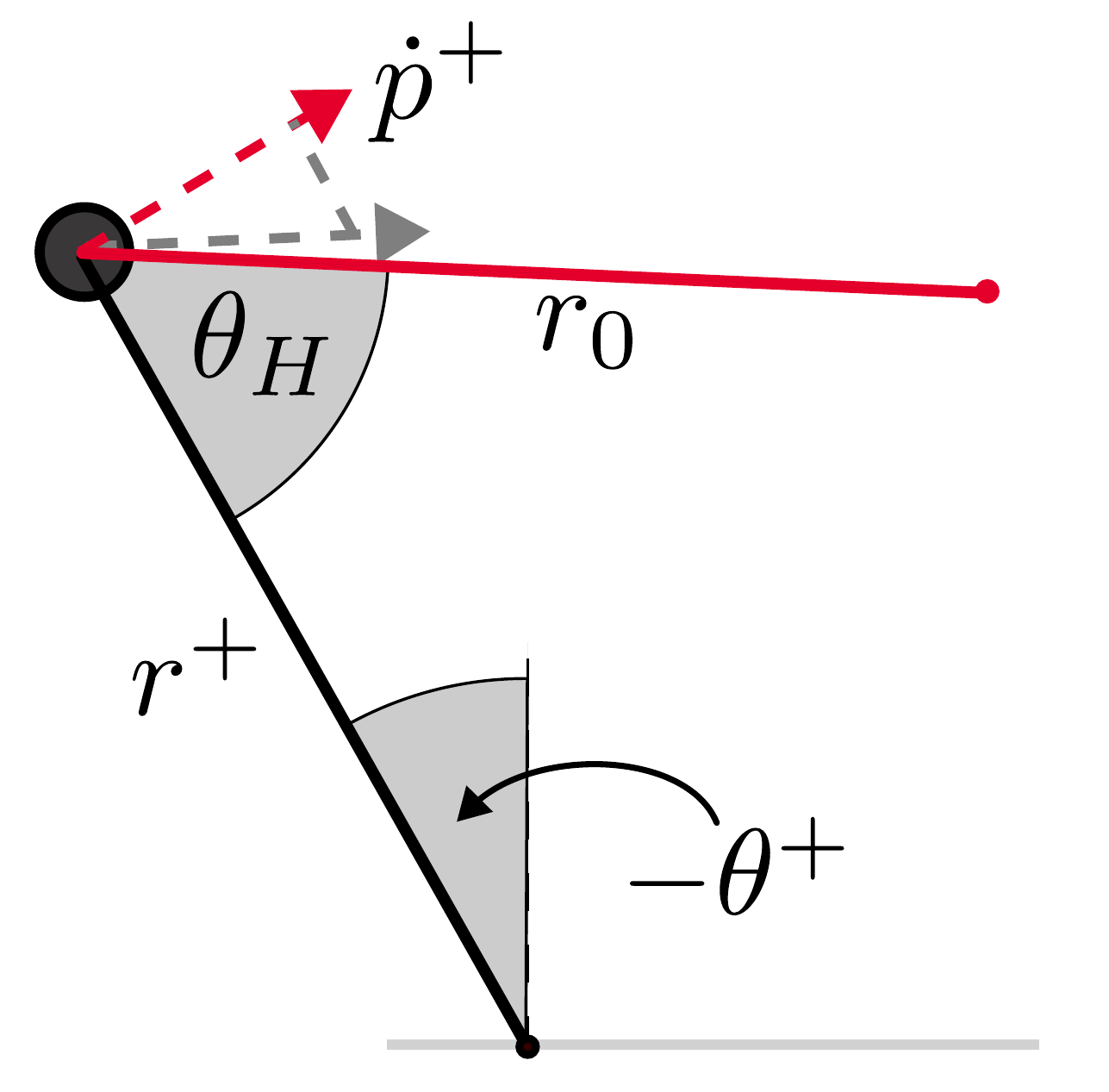}
        \caption{}
        \label{fig:postimpact}
    \end{subfigure}
    \caption{Simplified Bipedal Walker illustrated for the (a) continuous domain, (b) pre-impact, and (c) post-impact}
    \label{fig:model}
\end{figure}

\subsection{Planar Bipedal Model}
\subsubsection{Dynamics}
\label{sec:BipedDynamics}
We model the bipedal walker as an inverted pendulum, with mass $m$ and length $r$, at angle $\theta$ from the vertical. An illustration of this model is provided in Figure \ref{fig:model}.
Over one step, one leg (the stance leg) is pinned to the ground, and the other leg (the swing leg) is held at a constant length $r_0$ and constant angle $\theta_H$ from the stance leg.
The robot is capable of applying a force $u$ along its stance leg.

Under these assumptions, the robot state $[r, \theta, \dot{r}, \dot{\theta}]$ evolves according to the dynamics
\begin{gather}
 \begin{bmatrix}
        \ddot{r} \\
        \ddot{\theta}
    \end{bmatrix} =
        \begin{bmatrix}
         (mr\dot{\theta}^2 +  u - mg\cos{\theta})/m \\
        -(2\dot{r}\dot{\theta} - g\sin{\theta})/r
    \end{bmatrix}.
\end{gather}

When the swing leg contacts the ground, an inelastic collision occurs and a stance change is triggered.
All momentum in the direction of the swing leg is removed, a controlled impulse of $v$ is applied in the direction of the stance leg, and the swing leg is relabeled as the new stance leg.
The controlled impulse models the effect of action applied by the old stance leg during a short dual-stance phase, and can be used to recover the energy lost from stepping as well as stabilize the step-to-step dynamics.

The switching region $\mathcal{S}$ of this system can be described in terms of the position variables:
\begin{gather}
    S = \{r, \theta: r\cos(\theta) + r_0\cos(\pi + \theta_H - \theta) = 0\}.
\end{gather}
The reset map $\Delta$ maps the pre-impact states to post-impact states.
The post-impact positions are fixed by the geometry of the system at impact:
\begin{gather}
    \begin{bmatrix}
        r^+\\\theta^+
    \end{bmatrix} 
    = 
    \begin{bmatrix}
        r_0 \\ \theta_H - \theta
    \end{bmatrix}.
\end{gather}
Applying the inelastic collision, controlled impulse, and relabeling results in the post-impact velocities
\begin{gather}
    \begin{bmatrix}
        r\dot{\theta}^+\\
        \dot{r}^+
    \end{bmatrix} = 
    R(-\theta)
    \begin{bmatrix}
        v\sin(\theta) + [r\dot{\theta}\cos(\theta_H) + \dot{r}\sin(\theta_H)]\cos(\beta) \\
        v\cos(\theta) - [r\dot{\theta}\cos(\theta_H) + \dot{r}\sin(\theta_H)]\sin(\beta)
    \end{bmatrix},
\end{gather}
where $\beta=\theta-\theta_H$ and $R\in SO(2)$ is a 2D rotation matrix.

\subsubsection{Transformed System Dynamics}
To simplify the computation of intersections with the guard surface, we transform our coordinates such that the guard surface becomes a hyperplane.
This is done through a map $\phi: \mathbb{R}^4 \to \mathbb{R}^4$.
We define a transformed coordinate system 
\begin{gather}
x = \begin{bmatrix}
r, \tan{\theta}, \dot{r}, \frac{\dot{\theta}}{\cos^2 \theta}
\end{bmatrix}^T.
\end{gather}
Under these coordinates, the guard surface is defined as
\begin{gather}
    \mathcal{S}_x = \left\{x : 
    \begin{bmatrix}
        \frac{-1}{r_0 \sin \theta_H} & 1 & 0 & 0
    \end{bmatrix}
    x + \tfrac{1}{\tan \theta_H} = 0
    \right\}.
\end{gather}

The reset function can be similarly defined as $x^+=\Delta_x(x^-) = \phi(\Delta(\phi^{-1}(x^-))$.
We note that the singularities of the $\phi$ map are at $\pm \pi/2$, which is outside the operating envelope of the system.

\begin{figure}
         \centering
         \includegraphics[width=0.6\columnwidth]{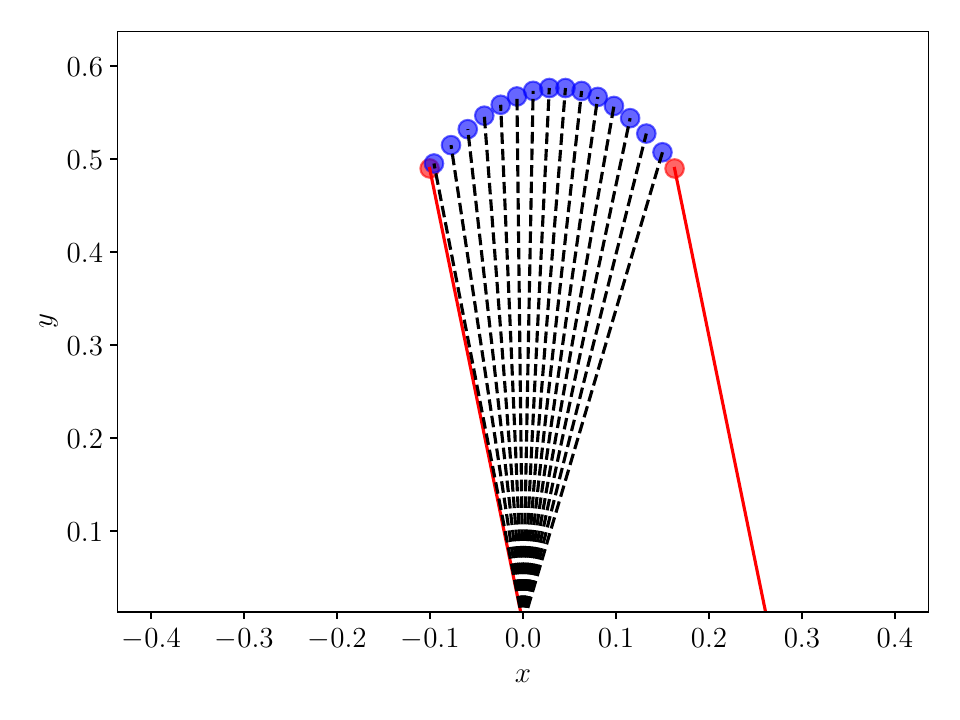}
         \label{fig:1step_traj_robot}
    \caption{Stable periodic limit cycle for the walker obtained from trajectory optimization.}
    \label{fig:trajopt}
\end{figure}
\subsubsection{Open-Loop Trajectory Optimization}
\label{sec:TrajOpt}
To realize walking for our model, we synthesize a periodic trajectory and stabilizing tracking controller.
Our trajectory optimization is formulated to choose an initial state $x_0^d$, feed-forward control sequence $u^{\text{ff}}_{i} ~\forall i \in 1..N-1$, and nominal dual-stance impulse $v^{\text{ff}}$.
From these specifications we can also compute a desired trajectory $x^d_i$.
It is transcribed as the single-shooting trajectory optimization problem
\begin{align}
    & \min_{x^d_0, u^{\text{ff}}_i, v^{\text{ff}}} && \sum_{i=0}^N (u^{\text{ff}}_i)^2 \\
    &~~~ \textrm{s.t.} && x^d_{i+1} = x^d_i + hf(x^d_i, u^{\text{ff}}_i), &&& \forall i \in 0,...,N-1,\notag \\
    & && \Delta(x^d_N, v^{\text{ff}}) = x^d_0, &&& \notag \\
    & && r_i \geq r_0, &&& \forall i \in 0,...,N, \notag \\
    & && x^d_N \in \mathcal{S}, \quad x^d_i \notin \mathcal{S}, &&& \forall i \in 0..N-1, \notag 
\end{align}
where $h$ is the discretization interval of the trajectory, and $N$ is the number of samples of the trajectory.
The problem is solved using the IPOPT solver through the \texttt{cyipopt} library \cite{moore2025cyipopt}.
The resulting nominal trajectory is shown in Figure \ref{fig:trajopt}.

\subsubsection{Step-to-Step Control}
\label{sec:step_control}

To stabilize the step-to-step dynamics, we introduce a discrete-time controller which adjusts the dual-stance impulse $v$.
For the open-loop trajectory computed in \ref{sec:TrajOpt}, we numerically estimate the step-to-step dynamics by flowing the continuous dynamics and reset map from \ref{sec:BipedDynamics} over one step, starting at perturbed initial conditions $x_0$ and perturbed impulse input $v$.
From the difference in initial and final points, we can approximate the step-to-step dynamics as
\begin{gather}
    x[k+1] - x^d_N = A(x[k]-x^d_N) + B(v^d + \delta v),
\end{gather}
where $\delta v$ is an additional control input term which adjusts $v$ from the nominal quantity $v^d$.
To stabilize this system, we select a controller $\delta v = K_{ds}(x^-_k-x^d_N)$ to place the eigenvalues of $F$ within the unit circle.

\subsubsection{Tracking Control}
\label{sec:tracking_control}
The planar bipedal walker can be stabilized onto a periodic gait using only the step-to-step controller outlined in Sec. \ref{sec:step_control}.
This control strategy is brittle and yields a relatively small invariant tube.
Thus, we develop a tracking controller designed to increase the size of the hybrid invariant tube by leveraging the automatic-differentiation capabilities of our framework.
We choose to develop controllers of the form $u(t) = u_{\text{ff}}(t) + K[x(t)-x^d(t)]$ and select the control gain $K$.

\begin{algorithm}
\caption{Reachability-based Control Design Procedure}
\label{alg:control_update}
\begin{algorithmic}
\Require Initial shape $\alpha_0$, Rescale tolerance $s_{\text{min}}$, Gradient Step Size $\eta$, Gradient Step Count $N$
\State $K^*\gets [0, 0, 0, 0]$
\State $s^* \gets \infty$
\State $\alpha^*\gets\alpha_0$
\While {$s^* > s_{\text{min}}$} 
    \For {$i = 1,...,N$} 
        \State $K^* \gets K^* - \eta \frac{\partial}{\partial K^*} \Phi'(\alpha^*, K^*)$
    \EndFor
    \State $s^* \gets  \max \quad s \quad  \text{s.t. } \Phi'(\alpha^*/s, K^*) < 1$
    \State $\alpha^* \gets \alpha^*/s$
\EndWhile
\State \Return $\alpha^*, K^*$
\end{algorithmic}
\end{algorithm}

First, we define the cost function $\Phi':GL(n) \times \mathbb{R}^n\rightarrow\mathbb{R}$ which, analogous to the function $\Phi$ defined in Sec. \ref{sec:resizing}, maps an initial shaping matrix to the associated post-impact length, given a particular choice of control gains.
We optimize the control gain $K^*$ and shaping matrix $\alpha^*$ in a bilevel fashion.
First, we iteratively update $K$ using gradient descent to decrease the cost function while holding $\alpha^*$ constant.
The gradient of the cost function, $\frac{\partial}{\partial K^*}\Phi'(\alpha^*, K^*)$ is obtained using automatic differentiation.

After a fixed number of gradient steps, we hold $K^*$ constant and optimize over $\alpha^*$ by solving the optimization problem \eqref{eq:scaling_bisection}.
This algorithm continues until the rescaling step is unable to increase the size of the initial normotope by more than the user-defined minimum factor $s_{\text{min}}$.
This process is summarized in Algorithm \ref{alg:control_update}.
Note that one can also optimize this cost function with respect to the initial normotope shape $\alpha_0$, leading to potentially an even larger final invariant set.


\section{Results}
\subsection{Hybrid Invariant Tube}
We first analyze the resultant closed-loop system from applying our procedure to the autonomous system described in Sec.~\ref{sec:BipedDynamics}.
For the discrete control used to stabilize the step-to-step dynamics, we place the linearized system's poles at $z=[0.0, 0.1, 0.2, 0.3]$ which stabilize the hybrid system.

We apply our procedure in Sec.~\ref{sec:implementation} to generate an estimate of the hybrid invariant tube, without continuous control applied.
This tube is shown in Figure \ref{fig:reachable-set}.
We empirically verify the invariance of the tube by performing a Monte Carlo simulation.
We initialize 100 trajectories on the boundary of the initial verified normotope and simulate the system for 15 crossings of the guard surface.
As expected, all trajectories successfully remain within the verified invariant tube.

We then design a tracking controller using the procedure outlined in Sec.~\ref{sec:tracking_control}.
With gradient step count $N=20$, step size $\eta=0.5$, and rescale tolerance $s_{\text{min}}=1.01$, Algorithm \ref{alg:control_update} converges to a $K^*$ that increased the size of the invariant tube by a factor of 4.25, as illustrated in Figure \ref{fig:reachable-set}.
We again verify through Monte Carlo sampling that the resulting larger tube is indeed invariant.

\begin{figure}
    \centering
    \begin{subfigure}[b]{.45\textwidth}
        \includegraphics[width=\textwidth]{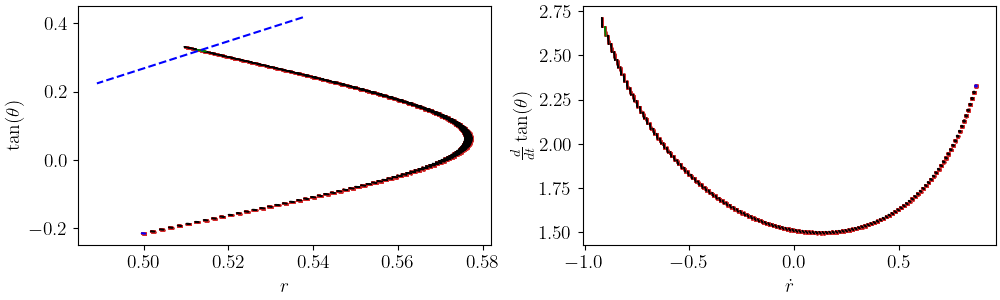}
        \caption{Hybrid Invariant Tube with no tracking control}
    \end{subfigure}
    \begin{subfigure}[b]{.45\textwidth}
        \includegraphics[width=\textwidth]{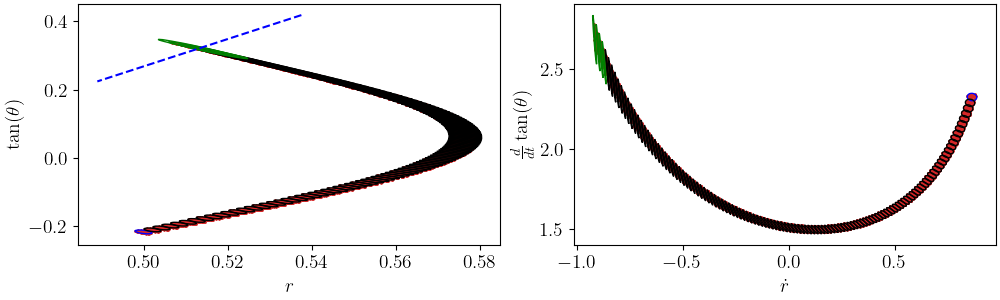}
        \caption{Hybrid Invariant Tube with tracking control included}
    \end{subfigure}
    \caption{Slices of the computed Hybrid Invariant Tube with no tracking control (above) and tracking control (below). The blue dotted line represents the guard surface $\mathcal{S}$. The green portions of the tube are the ones for which the tube intersects the guard surface.}
    \label{fig:reachable-set}
    \vspace{-5mm}
\end{figure}

\addtolength{\textheight}{-15mm}   

\subsection{Computational Efficiency}
A key benefit of our approach is the ability to quickly produce approximated forward-invariant sets.
To compute the reachable sets on our 4-dimensional system takes a total of 19.56 seconds, out of which 4.55 seconds are taken up by JIT compilation.
In contrast, the SOS programming approach taken in \cite{manchester2011transverse} takes 17.7 minutes to compute a reachable set for the 4-dimensional compass walker.
Similarly, the Hamilton-Jacobi-Bellman approach taken in \cite{choi2022computation} takes roughly 36 hours to compute a reachable set for the 4-dimensional compass walker.
While these approaches are not directly comparable, the speed at which we can compute approximate solutions is promising for embedding reachability analysis into control design or trajectory optimization algorithms.

\section{Limitations and Future Work}
One key assumption of this approach is the requirement that the guard surface is linear with respect to the choice of coordinates.
In this work, we found a map $\phi$ which converted the planar bipedal walker's states into coordinates for which this is true.
This map may or may not exist for all hybrid systems that follow the structure in \ref{eq:ol_hybridsys}.
However, we observe that other bipedal walker models such as the compass-gait walker do have a guard surface that is linear in their coordinates.
Future work will investigate the conditions surrounding this coordinate map for other hybrid systems. 

Another limitation is the inherent conservativeness in our approach.
The interval-based linear inclusions that model the evolution of the normotope and the effect of the reset map result in a conservative overapproximation of the true forward set.
However, the approach developed in \cite{harapanahalli2025normotope} shows promise in reducing this conservatism through the introduction of a \textit{controlled embedding system} and an additional tube refinement step.
In future work, we plan to integrate this refinement step into our invariant set computation pipeline.


Finally, our work only proves forward invariance of the identified hybrid tube.
While we know empirically that this set also yields asymptotic stability to a stable periodic orbit, this has yet to be proven.
Despite these limitations, we believe the computational speed and differentiability show promise for more complicated control design problems.



\section{Conclusion}
Our work presents a set-based reachability framework for efficiently estimating the region of attraction of a hybrid system.
We present a condition under which the reachable set of a periodic hybrid system is forward-invariant, and provide a numerical algorithm for computing it.
This condition is differentiable with respect to the initial set and system parameters, permitting reachability-guided control design.
We demonstrate its effectiveness by designing a controller and certifying the forward-invariant set of a planar walker model.
While this demonstration is limited in scale, our approach promises to extend to high-dimensional systems by leveraging interval analysis as well as the parallelization and auto-differentiation capabilities of JAX.
As such, future work will extend this methodology to more complex systems in order to fully demonstrate its benefits.

\bibliographystyle{ieeetr}
\bibliography{references.bib}

\end{document}